\documentclass[a4paper,12pt]{article}
\usepackage{amsmath}
\usepackage{amssymb}
\usepackage{amsthm}
\usepackage{amsfonts}
\usepackage{mathrsfs}
\usepackage{graphicx}
\usepackage{hyperref}
\usepackage{float}
\usepackage{slashed}
\usepackage{cite}
\usepackage{xcolor}
\usepackage{authblk}

\newtheorem{theorem}{Theorem}
\newtheorem{lemma}{Lemma}

\numberwithin{equation}{section}

\newtheoremstyle{named}{}{}{\itshape}{}{\bfseries}{.}{.5em}{\thmnote{#3 }#1}
\theoremstyle{named}

\hypersetup{ pdftitle={Report 2},
	pdfauthor={Mulyanto}, 
	pdfkeywords={}, 
	bookmarksnumbered, pdfstartview={FitH}, urlcolor=blue, }

\begin{document}
	\pagestyle{plain}

\title{\LARGE\textbf{Positive Mass in Scalar-Torsion Holography}}

\author[1]{Mulyanto}
\author[1]{Hadyan L. Prihadi}
\author[1]{Emir S. Fadhilla}
\author[1]{Ardian N. Atmaja}

\affil[1]{\textit{\small Research Center for Quantum Physics, National Research and Innovation Agency (BRIN),
Kompleks PUSPIPTEK Serpong, Tangerang 15310, Indonesia}}

\maketitle

\begin{abstract}
    We investigate the holographic renormalization of scalar-torsion gravity in a four-dimensional bulk spacetime with non-minimal derivative coupling. The asymptotic behavior of the static equations leads to an anti-de Sitter geometry for negative cosmological constants, allowing for a holographic interpretation via the AdS/CFT correspondence. The existence of unique solutions throughout the bulk is addressed. We study the effect of the non-minimal coupling parameter on the conformal dimension and the expectation value of the dual scalar operator, showing that the effective bulk mass can be tuned through the non-minimal coupling. Our results provide a formalism for finding non-tachyonic bulk scalar fields which vanish at the boundary.
\end{abstract}
\section{Introduction}
The search for a deeper understanding of the fundamental nature of gravity continues to inspire the development of theories that build upon or deviate from the principles of General Relativity (GR). Among these, teleparallel gravity and its generalizations have gained considerable attention as potential frameworks for describing gravitational interaction. Unlike GR, which attributes gravity to spacetime curvature, teleparallel theories formulate gravity in terms of torsion, using the Weitzenböck connection that is curvature-free but torsion-rich \cite{haya, arcos, maluf}. Teleparallel gravity theories have attracted significant interest over the years, mainly because of their potential applications in cosmological studies \cite{Bengochea:2008gz, Linder:2010py, Cai:2015emx}. One particularly compelling extension of teleparallel gravity is scalar-torsion theory. A common feature of these theories is the presence of one or more scalar fields. These scalar fields are typically coupled with the spacetime metric. This coupling is typically non-minimal \cite{hofman}.

One of the main challenges in generalized gravity theories is verifying whether they reduce to General Relativity (GR) under suitable conditions. Recent studies \cite{far2018, Habib:2018} have demonstrated that in teleparallel gravity, gravitational waves possess the same polarization modes as those predicted by GR, assuming the boundary term is minimally coupled to both the torsion scalar and the scalar field. Similarly, in the context of $f(T)$ gravity, the behavior of gravitational waves remains aligned with the predictions of Einstein’s theory \cite{bamba2013, capo2020}. Therefore, this theory should also demonstrate consistency, similar to General Relativity, when examined through other perspectives, such as the AdS/CFT correspondence.

While the theory has yielded significant insights in cosmological applications, investigating static solutions—especially black hole configurations—is equally crucial for assessing its consistency with General Relativity. Numerous works have established the existence of static, spherically symmetric black hole solutions in four-dimensional spacetime under this theoretical framework \cite{Wang:2011xf, Boehmer:2011gw, Ferraro:2011ks}. In \cite{kofinas2015torsi}, a broader class of black hole solutions was obtained for particular values of integration constants, featuring asymptotic hyperscaling behavior that breaks Lifshitz spatial symmetry and possesses a spherical horizon topology. At spatial infinity, the solution also approaches an asymptotically Anti-de Sitter (AdS) spacetime. However, the findings in \cite{kofinas2015torsi} are questionable, as they rely on the same master equation used in \cite{kofinas2012torsi}, which was later shown to contain a computational error, as demonstrated in \cite{Yaqin2017comment}. Nevertheless, the result is still noteworthy because scalar torsion theory asymptotically approaches AdS spacetime. This allowed us to make a comparison with general relativity within the context of AdS/CFT correspondence. 

The AdS/CFT correspondence \cite{Maldacena1999, GKP1998, Witten1998} tells us that an asymptotically anti-de Sitter theory has a conformal field theory dual in its boundary. A scalar field $\phi$ in the bulk is a source to the dual scalar field $\mathcal{O}$ in the CFT at the boundary. This source generates a deformation in the boundary theory. The response of this deformation is encoded in the 1-point function $\langle\mathcal{O}\rangle$ \cite{deHaro2001}. The partition function $Z_{\text{CFT}}$ of the CFT theory is equal to the classical partition function $Z_{\text{AdS}}$ of the bulk gravitational theory in a semi-classical limit. In this paper, we demonstrate that the scalar torsion theory with nonminimal derivative couplings asymptotically approaches AdS spacetimes, allowing for the implementation of holographic renormalization. We calculate the 1-point function $\langle\mathcal{O}\rangle$ using the renormalized bulk action with an asymptotically AdS gravitational solution. We also see how this 1-point function gets influenced by the bulk non-minimal derivative coupling.

The structure of this paper is organized as follows. Section \ref{static} is dedicated to static spacetime solutions, where we demonstrate that the solution approaches the AdS spacetime at spatial infinity. We also address the existence of a unique solution to this theory. In Section \ref{CFT}, we show the CFT dual for the theory and demonstrate the holographic renormalization. Finally, Section \ref{sec:conclusion} provides a summary of our results and suggests potential directions for future research. The local and global existence of the scalar torsion theories with non-minimal derivative couplings are also addressed in Appendix Section \ref{sec:localglobalex}.

\section{Static Spacetimes in Scalar Torsion Theory}\label{static}
The action for the scalar-torsion theory involving a non-minimal derivative coupling and including a cosmological constant $\Lambda$ is given by:
 \begin{eqnarray}\label{S}
     S=-\frac{1}{2\kappa ^{2}}\int{{{d}^{4}}}x~e(T+2\Lambda)-\int{{{d}^{4}}}x~e\left[\left( \frac{1}{2}-\xi T \right){{g}^{\mu \nu }}{{\partial }_{\mu }}\phi {{\partial }_{\nu }}\phi +V\right]~.
 \end{eqnarray}
Here, $e$ represents the determinant of the vielbein, $T$ denotes the torsion as define in \eqref{torsiondef} and $\xi$ denotes the non-minimal coupling parameter. The constant $\kappa \equiv \frac{1}{M_p}$, where $M_p$ is the four-dimensional Planck mass. By varying the action \eqref{S} with respect to the vielbein $e_a$ and the scalar field $\phi$, we obtain the corresponding equations of motion as follows: 
\begin{eqnarray}\label{eom1}
  && \left( \frac{2}{\kappa^{2}}-4\xi {{\phi }_{,\rho }}{{\phi }^{,\rho }} \right)\left[ (e{{S}_{\kappa }}^{\lambda \nu }{{e}_{{\bar{b}}}}^{\kappa }){{,}_{\nu }}{{e}^{{\bar{b}}}}_{\mu }+e\left( \frac{1}{4}T\delta _{\mu }^{\lambda }-{{S}^{\nu \kappa \lambda }}{{T}_{\nu \kappa \mu }}\right) \right] \notag\\ 
 &&+ e\frac{\Lambda}{\kappa^2}\delta^\lambda_\mu +4\xi \left[ \frac{1}{2}eT{{\phi }_{,\mu }}{{\phi }^{,\lambda }}+e{{S}_{\mu }}^{\nu \lambda }({{\phi }_{,\kappa }}{{\phi }^{,\kappa }})_{,\nu}  \right] \notag\\ 
 && +e\left( \frac{1}{2}{{\phi }_{,\rho }}{{\phi }^{,\rho }}\delta _{\mu }^{\lambda }-{{\phi }_{,\mu }}{{\phi }^{,\lambda }}+V\delta _{\mu }^{\lambda } \right)=0 ~,
\end{eqnarray}

\begin{eqnarray}\label{eom2}
{{\left[ e(1-2\xi T){{\phi }^{,\mu }} \right]}_{,\mu }}-e\frac{dV}{d\phi }=0~.
\end{eqnarray}
Note that this theoretical framework does not rely on curvature. As a result, the solutions to the equations \eqref{eom1} and \eqref{eom2} are characterized by a spacetime geometry defined by the metric tensor $g_{\mu \nu} = \eta_{ab} e_{\mu}^{a} e_{\nu}^{b}$, which can lead to physical predictions distinct from those of standard general relativity. In this context, it is particularly interesting to study the simplest class of solutions, namely static spacetimes, in order to verify the viability of this alternative gravitational theory.

We begin by considering the following ansatz for the metric:
\begin{eqnarray}\label{metrik}
    d{{s}^{2}}=-N{{(r)}^{2}}d{{t}^{2}}+K{{(r)}^{-2}}d{{r}^{2}}+R{{(r)}^{2}}d{{\Omega }^{2}}~.
\end{eqnarray}
The choice of metric as given above leads to the following form for the associated vielbein corresponding to metric \eqref{metrik}
\begin{eqnarray}
    {{e}^{a}}_{\mu }=(N(r),K{{(r)}^{-1}},R(r){{\hat{e}}^{{\bar{b}}}}_{i})~.
\end{eqnarray}
 As a result, the torsion tensor in this case takes the form
 \begin{eqnarray}\label{torsion}
     T=-2{{K}^{2}}\frac{{{R}'}}{R}\left[ \frac{{{R}'}}{R}+2\frac{{{N}'}}{N} \right]~.
 \end{eqnarray}
By substituting equations \eqref{metrik} and \eqref{torsion} into \eqref{eom1}, and assuming a radial dependence for the scalar field, $\phi = \phi(r)$, we obtain the following set of equations
\begin{eqnarray}\label{eq1}
    &&2\left( \frac{1}{\kappa ^{2}{{K}^{2}}}-2\xi {{{{\phi }'}}^{2}} \right)\left[ 2\frac{{R}'{K}'}{RK}+\frac{{{{{R}'}}^{2}}}{{{R}^{2}}} \right.\left. +2\frac{{{R}''}}{R}-\frac{1}{{{R}^{2}}{{K}^{2}}} \right]\notag\\
   && -16\xi {\phi }'\frac{{{R}'}}{R}\left( \frac{{{K}'}}{K}{\phi }'+{\phi }'' \right)+\frac{1}{{{K}^{2}}}\left( {{{{\phi }'}}^{2}}+\frac{2}{{{K}^{2}}}\left(V+\frac{\Lambda}{\kappa^2}\right) \right)=0~,
\end{eqnarray}
 \begin{eqnarray}\label{eq2}
  & &\left[ \frac{{{R}'}}{R}\left( \frac{{{R}'}}{R}+2\frac{{{N}'}}{N} \right)-\frac{1}{{{R}^{2}}{{K}^{2}}} \right]2\left( 2\xi {{{{\phi }'}}^{2}}-\frac{1}{\kappa ^{2}{{K}^{2}}} \right) \notag\\ 
 && +\left( \frac{{{R}'}}{R}\left( \frac{{{R}'}}{R}+\frac{2{N}'}{N} \right) \right)8\xi {{{{\phi }'}}^{2}}+\frac{1}{{{K}^{2}}}\left( {{{{\phi }'}}^{2}}-\frac{2}{{{K}^{2}}}\left(V+\frac{\Lambda}{\kappa^2}\right) \right)=0~, 
 \end{eqnarray}
\begin{eqnarray}\label{eq3}
  && 2\left( \frac{1}{\kappa ^{2}{{K}^{2}}}-2\xi {{{{\phi }'}}^{2}} \right)\left[ \frac{{{N}'}}{N}\left( \frac{{{R}'}}{R}+\frac{{{K}'}}{K} \right)+\frac{{{R}''}}{R}+\frac{{R}'{K}'}{RK}+\frac{{{N}''}}{N} \right] \notag\\ 
 && -8\xi {\phi }'\left( \frac{{{R}'}}{R}+\frac{{{N}'}}{N} \right)\left( \frac{{{K}'}}{K}{\phi }'+{\phi }'' \right)+\frac{1}{{{K}^{2}}}\left( {{{{\phi }'}}^{2}}+\frac{2}{{{K}^{2}}}\left(V+\frac{\Lambda}{\kappa^2}\right) \right)=0 ~,
\end{eqnarray}
\begin{eqnarray}\label{eq4}
    {\phi }'\left( \frac{{{K}'}}{K}{\phi }'+{\phi }'' \right)=0~.
\end{eqnarray}
In a similar manner, equation \eqref{eom2} in this context becomes
\begin{eqnarray}\label{eq5}
    {\left\{ \left[ 1+4\xi \left( \frac{{{K}^{2}}{R}'}{R}\left[ \frac{{{R}'}}{R}+\frac{2{N}'}{N} \right] \right) \right]  \times N{{R}^{2}}K{\phi }' \right\}}^{\prime }-\frac{N{{R}^{2}}}{K}\frac{dV}{d\phi }=0~.
\end{eqnarray}
The set of field equations (\ref{eq1}-\ref{eq5}) is the system of equations for the metric functions and scalar field within this model, which needs to be solved simultaneously. 

\subsection{Asymptotic Behaviour}\label{asymptoticBehav}
The behaviour of the scalar field, \(\phi\), at the asymptotic boundary, \(r\rightarrow\infty\), is assumed to go to vacuum, i.e. \(\lim_{r\rightarrow\infty}\phi=\phi_{\text{vac}}\) such that \(V(\phi_{\text{vac}})=0\). As such, the resulting Einstein's equations are
\begin{eqnarray}
    &&\left[ 2\frac{{R}'{K}'}{RK}+\frac{{{{{R}'}}^{2}}}{{{R}^{2}}} \right.\left. +2\frac{{{R}''}}{R}-\frac{1}{{{R}^{2}}{{K}^{2}}} \right]+\frac{\Lambda}{{{K}^{2}}}=0,\\
  &&\left[ \frac{{{R}'}}{R}\left( \frac{{{R}'}}{R}+2\frac{{{N}'}}{N} \right)-\frac{1}{{{R}^{2}}{{K}^{2}}} \right]+\frac{\Lambda}{{{K}^{2}}} =0 ,\\
  && \frac{{{N}'}}{N}\left( \frac{{{R}'}}{R}+\frac{{{K}'}}{K} \right)+\frac{{{R}''}}{R}+\frac{{R}'{K}'}{RK}+\frac{{{N}''}}{N}   + \frac{\Lambda}{K^2}=0 .
\end{eqnarray}
The function \(R\) can chosen to be \(r\) by transformation of the radial coordinate, without changing the geometry of the manifold. Let us consider \eqref{eq1}, with \(R=r\). We get
 \begin{equation}
      (K^2)'r+K^2 -1+\Lambda r^2=0,
 \end{equation}
 whose solution is given by
 \begin{equation}
     K^2(r)=1-\frac{\Lambda r^2}{3}.
 \end{equation}
 For \(\Lambda<0\) the vacuum solution is AdS. Now, consider \eqref{eq3} which reads,
 \begin{equation}
     {N}'\left(3-2\Lambda r^2\right) +r{N}''\left(3-\Lambda r^2\right)  + 2\Lambda rN=0,
 \end{equation}
 for \(R=r\). This left us with 
 \begin{equation}
 N^2(r)= 1-\frac{\Lambda r^2}{3},
 \end{equation}
 after substitution of solution for \(K(r)\) and using the Boundary condition from \eqref{eq2}. In conclusion, at \(r\rightarrow\infty\) the metric behaves as
 \begin{equation}\label{solution1}
     ds^2\approx-\left(1-\frac{\Lambda r^2}{3}\right)d{{t}^{2}}+\frac{d{{r}^{2}}}{1-\frac{\Lambda r^2}{3}}+r^2d{{\Omega }^{2}}.
 \end{equation}
 We summarize this behaviour of the metric at the asymptotic boundary as follows: If \(\lim_{r\rightarrow\infty}\phi=\phi_{\text{vac}}\) such that \(V(\phi_{\text{vac}})=0\), then the spacetime, whose metric satisfies (\ref{eq1}-\ref{eq4}),  is asymptotically flat for \(\Lambda=0\), or Anti-de Sitter for \(\Lambda<0\).



\subsection{Existence of Unique Solutions }\label{masterequ}

Among the equations \eqref{eq1}-\eqref{eq5}, it is expected that we impose a constraint on the theory. In this case we choose the equation \eqref{eq4}, which can be expressed as
 \begin{eqnarray}
     {\phi }'=\frac{\nu }{K}
 \end{eqnarray}
where $\nu $ is a constant. We then introduce new variables $x, y, z$ as well as $X, Y, Z$, detailed calculations can be found in Appendix \ref{sec:localglobalex}. 
Through a series of computations, the field equations can be simplified to a highly nonlinear ordinary differential equation, which we will henceforth call the master equation for $Y(x)$, 
\begin{eqnarray}\label{mastereq}
  && \frac{{{d}^{2}}Y}{d{{x}^{2}}}+Y\frac{dY}{dx}\left\{ -\frac{3}{2}-\frac{3\tilde{\eta }}{4\hat{\eta }}+\frac{6}{\tilde{\eta }\hat{\eta }} \right\}+\frac{dY}{dx}\frac{{{e}^{-2x}}}{{{\nu }^{2}}}\left\{ \frac{1}{4}-\frac{{\hat{\eta }}}{{\tilde{\eta }}} \right\}+{{\left( \frac{dY}{dx} \right)}^{2}}\left\{ -\frac{1}{2}+\frac{{\hat{\eta }}}{{\tilde{\eta }}} \right\} \notag\\ 
 && +3\frac{dY}{dx}+{{Y}^{2}}\left\{ -\frac{9\tilde{\eta }}{4\hat{\eta }}+\frac{9\hat{\eta }}{{\tilde{\eta }}} \right\}+\frac{{{e}^{-2x}}}{{{\nu }^{2}}}Y\left\{ \frac{{\tilde{\eta }}}{{\hat{\eta }}}-\frac{3{{\nu }^{2}}\hat{\eta }}{{\tilde{\eta }}} \right\}+\frac{{{e}^{-4x}}}{{{\nu }^{2}}}\frac{{\hat{\eta }}}{4\tilde{\eta }}+\frac{{{e}^{-2x}}}{{{\nu }^{2}}}=0 \notag\\
\end{eqnarray}
where $\hat{\eta }$ and $\tilde{\eta}$ are constants defined as:
\begin{eqnarray}
  \hat{\eta }&=& 2\xi {{\nu }^{2}}{{\kappa }^{2}}-1 ~,\notag\\ 
  \tilde{\eta }&=&6\xi {{\nu }^{2}}{{\kappa }^{2}}-1  ~.
\end{eqnarray}
It can be demonstrated that the scalar-torsion gravity theory described by the master equation \eqref{mastereq}, which incorporates both non-minimal derivative couplings and a cosmological constant term, admits unique local and global solutions (see Appendix \ref{sec:localglobalex} for details calculation).

\section{CFT Dual of Scalar Torsion Theory with Non-minimal Coupling}\label{CFT}
As seen in subsection \ref{asymptoticBehav}, the asymptotic behavior of the scalar-torsion theory approach Anti-de Sitter solution for $\Lambda<0$. Therefore, we might expect that this theory has a CFT dual in its asymptotic boundary. In the context of the AdS/CFT correspondence, a scalar field $\phi$ that propagates in the bulk corresponds to a source $\phi^{(0)}$ for a deformation in the boundary theory which gives a response $\langle\mathcal{O}\rangle$, were $\mathcal{O}$ is a scalar operator in the CFT theory. This generates a deformed action in the boundary in the form of
\begin{equation}
    \delta S_{\text{CFT}}=\int_{r\rightarrow \infty}d^3x e\phi^{(0)}\langle\mathcal{O}\rangle.
\end{equation}
In this work, we are interested in investigating how the non-minimal coupling $\xi$ affect the expectation value $\langle\mathcal{O}\rangle$ and its conformal dimension in the boundary theory. In this section, we aim to calculate the expectation value $\langle\mathcal{O}\rangle$ from
\begin{equation}\label{1pt1}
    \langle\mathcal{O}\rangle\sim\frac{\delta S_{\text{ren.}}}{\delta \phi^{(0)}},
\end{equation}
where $S_{\text{ren.}}$ is a renormalized action for the scalar field.\\
\indent In calculating the renormalized action, we need to perform the holographic renormalization procedure for the scalar field. The method is already well known. We first consider the asymptotically-AdS solution as $r\rightarrow \infty$ in the form of
\begin{equation}
    ds^2=-r^2dt^2+\frac{dr^2}{r^2}+r^2d\Omega^2,
\end{equation}
where we use the solution in eq. \eqref{solution1} in $r\rightarrow\infty$ limit and we consider $\Lambda=-3/l^2$ with $l$ be the AdS radius. We also set $l=1$ for simplicity. Using $z=\frac{1}{r}$, the solution becomes
\begin{equation}
    ds^2=\frac{1}{z^2}(-dt^2+dz^2+d\Omega^2),
\end{equation}
and now the asymptotic boundary is located at $z=0$. This is a $AdS_4$ spacetime in the Poincar\'e coordinates.\\ 
\indent In calculating holographic renormalization, the coordinate $\rho=z^2$ is preferred. In this case, the solution becomes
\begin{equation}
    ds^2=\frac{d\rho^2}{4\rho^2}+\frac{1}{\rho}(-dt^2+d\Omega^2).
\end{equation}
We want to solve the equation of motion for the scalar field $\phi$ near the asymptotic boundary where $\rho\rightarrow 0$. In this case, we solve the equation of motion perturbatively near $\rho =0$. The asymptotic solution that we are looking for has the form of
\begin{equation}
    \phi(\rho,t,\Omega)=\rho^{(3-\Delta)/2}\tilde{\varphi}(\rho,t,\Omega),
\end{equation}
where
\begin{equation}
\tilde{\varphi}(\rho,t,\Omega)=\phi^{(0)}+\rho\phi^{(2)}+\rho^2\phi^{(4)}+...\;.
\end{equation}
\indent After expressing all of the expansion modes $\phi^{(n)}$ with $n>0$ in terms of $\phi^{(0)}$, we can insert this solution back to the scalar field action and obtain the on-shell action. However, the on-shell action diverges in the boundary and needs to be renormalized by introducing a counterterm $S_{\text{ct}}$. We first regularized the on-shell action by inserting a cut-off $\varepsilon$ so that the asymptotic boundary is located at $\rho=\varepsilon$ and we obtain the regularized action $S_{\text{reg.}}$. The divergence terms in the $S_{\text{reg.}}$ is then cancelled by the counterterm $S_{\text{ct}}$. Finally, the renormalized action is obtained from
\begin{equation}
    S_{\text{ren.}}=\lim_{\varepsilon\rightarrow 0} (S_{\text{reg.}}+S_{\text{ct}}).
\end{equation}
\subsection{Asymptotic Scalar Equation of Motion}
The equation of motion for scalar field is initially given by eq. \eqref{eq5}. In the $(\rho,t,\Omega)$ coordinates, as $\rho\rightarrow 0$, the Torsion approaches constant value $T=-6$, which agrees with the anti-de Sitter solution with constant negative curvature. In this coordinate system, the equation of motion for the scalar field reads
\begin{equation}
    4(1+12\xi)\rho^{5/2}\partial_\rho(\rho^{-1/2}\partial_\rho\phi)+\rho(1+12\xi)\square\phi-m^2\phi=0,
\end{equation}
where we define
\begin{equation}
    \square\equiv -\frac{\partial^2}{\partial t^2}+\frac{1}{\sin\theta}\frac{\partial}{\partial\theta}\bigg(\sin\theta\frac{\partial}{\partial\theta}\bigg)+\frac{1}{\sin^2\theta}\frac{\partial^2}{\partial\varphi^2}.
\end{equation}
Furthermore, we are looking for a solution in the form of $\phi(\rho,t,\Omega)=\rho^{(3-\Delta)/2}\bar{\phi}(\rho,t,\Omega)$, following the standard holographic renormalization procedure in $AdS_4$ with
\begin{equation}
    \bar{\phi}(\rho,t,\Omega)=\phi^{(0)}(t,\Omega)+\rho\phi^{(2)}(t,\Omega)+\rho^2\phi^{(4)}+...\;.
\end{equation}
In this expression, the equation of motion becomes
\begin{align}
    \bar{\phi}&[\Delta(3-\Delta)(1+12\xi)+m^2]\\\nonumber
    &\;\;\;\;\;+\rho[2(1+12\xi)(2\Delta-5)\partial_\rho\bar{\phi}-(1+12\xi)\square\bar{\phi}-(1+12\xi)\rho\partial_\rho^2\bar{\phi}]=0.
\end{align}
\indent By setting $\rho=0$, we have
\begin{equation}
    \Delta(3-\Delta)(1+12\xi)+m^2=0.
\end{equation}
The positive solution of $\Delta$ gives us
\begin{equation}
    \Delta_\xi=\frac{3}{2}+\sqrt{\frac{9}{4}+m_{\text{eff}}(\xi)^2},
\end{equation}
where $m_{\text{eff}}(\xi)^2=\frac{m^2}{1+12\xi}$ is defined as the effective mass due to the presence of the non-minimal coupling constant $\xi$. Since $\xi$ can take negative value, the square of the effective mass $m_{\text{eff}}(\xi)^2$ can be negative even if we keep the mass of the bulk scalar field $\phi$ real ($m^2>0$). This is one of the main advantages of a theory with nonminimal derivative coupling.\\
\indent Near the asymptotic boundary, the scalar field solution can then be expected to behave as
\begin{equation}
    \phi(\rho,t,\Omega)=\rho^{(3-\Delta_\xi)/2}\phi^{(0)}(t,\Omega)+\rho^{\Delta_\xi/2}\phi^{(2\Delta_\xi-3)}(t,\Omega)+...\;.
\end{equation}
By ensuring $\lim_{\rho\rightarrow 0}\phi\rightarrow0$ in the boundary, we need $\Delta_\xi<3$. In this case, $m_{\text{eff}}(\xi)^2$ is required to be negative. Additionally, to keep $\Delta_\xi$ be real-valued and satisfy the Breitenlohner-Freedman bound, the range of $m_{\text{eff}}(\xi)^2$ is given by
\begin{equation}\label{massbound}
    -\frac{9}{4}<m_{\text{eff}}(\xi)^2<0.
\end{equation}
In a scalar theory without nonminimal derivative coupling $\xi$, we need to make the scalar field in the bulk slightly tachyonic with $m^2<0$. However, with an appropriate negative value of $\xi$, the mass of the bulk scalar field can still have a real and positive value.\\
\indent By plugging in the $\rho=0$ solution back to the equation of motion, we have
\begin{equation}
    2(2\Delta_\xi-5)\partial_\rho\bar{\phi}-\square\bar{\phi}-\rho\partial_\rho^2\bar{\phi}=0.
\end{equation}
This equation is identical to the standard scalar equation of motion in the holographic renormalization, with the conformal dimension $\Delta$ is replaced by $\Delta_\xi$. The recursion relation for $\phi^{(n)}$ is then given by
\begin{equation}
    \phi^{(2n)}=\frac{1}{2n(2\Delta_\xi-3-2n)}\square\phi^{(0)},\label{solutionscalarEOM}
\end{equation}
with $n\neq\Delta-3/2$.
\subsection{Holographic Renormalization and CFT 1-point Function}
From the solution in eq. \eqref{solutionscalarEOM}, the counterterm of the scalar action is given by \cite{deHaro2001}
\begin{equation}
    S_{\text{ct}}=(1+12\xi)\int d^3 x\sqrt{\gamma}\bigg(\frac{3-\Delta_\xi}{2}\phi^2+\frac{1}{2(2\Delta_\xi-3-2)}\phi\square_\gamma \phi+...\bigg).
\end{equation}
This is the standard counterterm for a massive scalar field $\phi$ in AdS$_4$/CFT$_3$, but with $\Delta$ replaced by $\Delta_\xi$. Furthermore, we also have an extra factor of $(1+12\xi)$ in front of the counterterm action. This extra factor also appears when we evaluate the bulk action at the asymptotic boundary. In this case, $\gamma$ is the induced metric in the boundary theory, which is given by a time coordinate $t$ and $S^2$ coordinates $\Omega$. This counterterm along with the original bulk action gives us the renormalized action of the theory. The 1-point function $\langle\mathcal{O}\rangle$ obtained by eq. \eqref{1pt1} is then given by
\begin{equation}\label{1pt2}
    \langle\mathcal{O}\rangle \sim -(1+12\xi)(2\Delta_\xi-3)\phi^{(2\Delta_\xi-3)}+C(\phi^{(0)}),
\end{equation}
where $C$ is a function of the boundary source term $\phi^{(0)}$. The first term of eq. \eqref{1pt2} is universal.\\
\indent From this result, we see that a four-dimensional scalar-torsion theory with a non-minimal derivative coupling is dual to a three-dimensional CFT $\mathcal{O}$ with a conformal dimension $\Delta_\xi$ that is controlled by the non-minimal derivative coupling $\xi$. The bulk massive scalar field $\phi$ behaves as a source $\phi^{(0)}$ to the scalar field $\mathcal{O}$ in the asymptotic boundary. \\
\indent The response $\langle\mathcal{O}\rangle$ of the source $\phi^{(0)}$ is given by $\phi^{(2\Delta_\xi-3)}$, up to a multiplicative factor that depends on the non-minimal derivative coupling. Due to the presence of $\xi$, this response can have different sign compared to $\xi=0$ case. To see this, when we require the effective mass $m_{\text{eff}}(\xi)^2$ to obey the bound in eq. \eqref{massbound} and assuming $m^2>0$, the value of $\xi$ needs to be negative. With $\Delta_\xi$ ranged between $\frac{3}{2}$ and $3$, the term $(2\Delta_\xi-3)$ is always positive. However, the term $(1+12\xi)$ must be negative since $\xi$ must be less than $-\frac{1}{12}$.



\section{Conclusion}\label{sec:conclusion}
In this work we have studied the scalar-torsion gravity theory with non-minimal derivative coupling and cosmological constant in a 3+1-dimensional bulk spacetime. By focusing on static, spherically symmetric solutions, we derived a set of highly nonlinear equations governing the system, called the master equations, which describe the dynamics of the scalar-torsion system. We have proved both the local and global existence of solutions to the master equation using the Picard iteration method and the contraction mapping theorem. These results confirm that the theory admits well-defined static solutions under physically reasonable assumptions.

We have also studied the asymptotic behavior of these solutions, showing that the spacetime approaches anti-de Sitter geometry for negative cosmological constants, $\Lambda < 0$. This naturally motivated a holographic interpretation based on the AdS/CFT correspondence. In this context, we analysed how the non-minimal coupling parameter $\xi$ modifies the dual conformal field theory at the boundary, in particular influencing the conformal dimension and the expectation value of the dual scalar operator. Using a holographic renormalization procedure, we show that the presence of the torsion scalar leads to a modified bulk scalar field equation where the effective mass explicitly depends on the coupling parameter $\xi$.
As a result, the conformal dimension of the dual operator can be continuously tuned, even allowing real conformal dimensions with real bulk mass for $\xi<-1/12$, where the mass satisfies
\begin{equation}
    -\frac{9(1+12\xi)}{4}>m^2>0.
\end{equation}

Overall, our findings confirm that scalar-torsion gravity theories offer a consistent extension of General Relativity and demonstrate that holographic renormalization is also present in teleparallel gravity. Future directions of this research include a deeper investigation into the thermodynamic aspects of these solutions, particularly within the framework of holographic entanglement entropy and black hole entropy in torsion-based AdS spacetimes. However, it is crucial to confirm the existence of black hole solutions within this theory.  This will be further discussed in our future work.

\appendix
	\section*{Appendix}
\section{Construction of Torsion Theory of Gravity}
We begin by defining the torsion tensor on a 4-dimensional spacetime manifold $\mathcal{M}^4$, which takes the form
	\begin{equation}\label{torsiondef}
	{T^{\lambda}}_{\mu\nu} = {\omega^{\lambda}}_{\nu\mu}-{\omega^{\lambda}}_{\mu\nu} ~ ,
	\end{equation}  
where ${\omega^\lambda}_{\mu\nu}$ denotes an alternative connection with spacetime indices $\mu, \nu, \lambda \in \{0,1,2,3\}$. The connection can be expressed in terms of the vielbein fields $e^a$ and its duals $e_a$ as
\begin{equation}\label{altconnviel}
\omega^\lambda_{\ \mu\nu} = e_a^{\ \lambda} \partial_\nu e^a_{\ \mu},
\end{equation}
which defines the Weitzenb\"ock connection. From this expression, we obtain the torsion tensor:
\begin{equation}\label{torsionviel}
T^\lambda_{\ \mu\nu} = -e_a^{\ \lambda}(\partial_\mu e^a_{\ \nu} - \partial_\nu e^a_{\ \mu}).
\end{equation}
The vielbein fields satisfy the orthonormality conditions:
\begin{equation}\label{identviel}
e^a_{\ \mu}e^{\ \nu}_a = \delta^\nu_\mu, \quad e^a_{\ \mu}e^{\ \mu}_b = \delta^a_b,
\end{equation}
where $a,b = 0,1,2,3$ denote the flat (tangent space) indices, while $\mu,\nu$ represent the spacetime indices. Here, we have the metric tensor 
	\begin{equation}\label{tensmet}
	g_{\mu\nu}=\eta_{ab}{e^a}_\mu{e^{b}}_\nu ~ ,
	\end{equation}
	where $ \eta_{ab}=\text{diag}(-1,1,\dots,1) $ is the Minkowski metric in the Lorentz frame.
	The main measurement of geometric deformation in the theory is the torsion tensor, which substitutes the Riemann tensor. 

We further define the contorsion tensor $K^\lambda_{\ \mu\nu}$ that relates the alternative connection $\omega^\lambda_{\ \mu\nu}$ to the Levi-Civita connection $\Gamma^\lambda_{\ \mu\nu}$ through:
	\begin{eqnarray}\label{contorsion}
	\mathcal{K}_{\lambda\mu\nu}&=&\frac{1}{2}(T_{\nu\lambda\mu}-T_{\mu\nu\lambda}-T_{\lambda\mu\nu})\notag\\
	&=&\omega_{\lambda\mu\nu} - \Gamma_{\lambda\mu\nu} ~ ,
	\end{eqnarray}
	where $\Gamma_{\lambda\mu\nu} = g_{\lambda\alpha} {\Gamma^{\alpha}}_{\mu\nu} $ and $g_{\lambda\alpha} $ is the spacetime metric endowed on ${\mathcal M}^4$. Within the teleparallel framework, a new tensor \( S^{\mu\nu\lambda} \) is introduced, which is given by
\begin{equation}\label{tensorS}
S^{\mu\nu\lambda} = \frac{1}{2} \mathcal{K}^{\nu\lambda\mu} + \frac{1}{2} \left( g^{\mu\lambda} {T_{\rho}}^{\rho\nu} - g^{\mu\nu} {T_{\rho}}^{\rho\lambda} \right) = -S^{\mu\lambda\nu} ~ ,
\end{equation}
This antisymmetric tensor plays a crucial role in the construction of the torsion scalar, as formulated in \cite{Habib:2018}
	\begin{eqnarray}\label{torsionscalar}
	T&=&S^{\mu\nu\lambda}T_{\mu\nu\lambda}\notag\\
	&=&\frac{1}{4}T^{\mu\nu\lambda}T_{\mu\nu\lambda}+\frac{1}{2}T^{\mu\nu\lambda}T_{\lambda\nu\mu}-{T_\nu}^{\nu\mu}{T^{\lambda}}_\lambda\mu ~ .
	\end{eqnarray}
In the framework of teleparallel gravity, the connection is subject to an additional constraint: it must be free of curvature \cite{kofinas2012torsi},
\begin{eqnarray}
  {\bar{\mathcal{R}}^{\rho}}_{~ \sigma\mu\nu}=\partial_{\mu}\omega^{\rho}_{~ \nu\sigma}-\partial_{\nu}\omega^{\rho}_{~ \mu\sigma}+\omega^{\rho}_{~ \mu\lambda}\omega^{\lambda}_{~ \nu\sigma}-\omega^{\rho}_{~ \nu\lambda}\omega^{\lambda}_{~ \mu\sigma}=0 ~ .
\end{eqnarray}
However, it is crucial to notice that the curvature constructed from the Levi-Civita connection \( {\Gamma^{\lambda}}_{\mu\nu} \) remains non-vanishing.
\begin{eqnarray}\label{Riemanncurv}
{{\mathcal{R}}^{\rho}}_{\sigma\mu\nu}&=&\partial_{\mu}\Gamma^{\rho}_{~ \nu\sigma}-\partial_{\nu}\Gamma^{\rho}_{~ \mu\sigma}+\Gamma^{\rho}_{~ \mu\lambda}\Gamma^{\lambda}_{~ \nu\sigma}-\Gamma^{\rho}_{~ \nu\lambda}\Gamma^{\lambda}_{~ \mu\sigma} \notag\\
&=& \nabla_\nu{\mathcal{K}^{\rho}}_{\mu\sigma}-\nabla_\mu{\mathcal{K}^{\rho}}_{\nu\sigma}+{\mathcal{K}^{\rho}}_{\nu\lambda}{\mathcal{K}^{\lambda}}_{\mu\sigma}-{\mathcal{K}^{\rho}}_{\mu\lambda}{\mathcal{K}^{\lambda}}_{\nu\sigma} ~ ,
\end{eqnarray}
where \( \nabla_\nu \) denotes the covariant derivative with respect to the spacetime metric \eqref{tensmet}. Contracting the indices in \eqref{Riemanncurv} yields the Ricci tensor
\begin{eqnarray}\label{Riccitens}
{\mathcal{R}}_{\mu\nu} = {\mathcal{R}}^{\rho}_{~ \mu\rho\nu} = \nabla_\nu{\mathcal{K}^{\rho}}_{\rho\mu}-\nabla_\rho{\mathcal{K}^{\rho}}_{\nu\mu}+{\mathcal{K}^{\rho}}_{\nu\lambda}{\mathcal{K}^{\lambda}}_{\rho\mu}-{\mathcal{K}^{\rho}}_{\rho\lambda}{\mathcal{K}^{\lambda}}_{\nu\mu} ~ .
\end{eqnarray}
The corresponding Ricci scalar is given by
\begin{eqnarray}\label{Riccisclr}
{\mathcal{R}} = g^{\mu\nu}{{\mathcal{R}}}_{\mu\nu} = -T+2\nabla_\mu {T_\nu}^{\nu\mu} ~ ,
\end{eqnarray}
which shows that the Ricci scalar of the manifold can be recast as a combination of the torsion scalar \( T \) and a total divergence term involving the torsion tensor. Thus, in this formulation of gravity, the geometrical and dynamical properties of spacetime are governed by torsion rather than curvature.

\section{Local and Global Existence of Unique Solutions}\label{sec:localglobalex}
We define new variables as in \cite{kofinas2012torsi, Yaqin2017comment}
	\begin{equation}\label{Y}
	Y\equiv y^2 ~ ,
	\end{equation}
	\begin{equation}\label{Z}
	Z\equiv z^2~ ,
	\end{equation}
	where $ x,y,z $ 
	\begin{equation}\label{x}
	x=\ln R ~ ,
	\end{equation}
	\begin{equation}\label{y}
	y=\frac{\dot{R}}{R} ~ ,
	\end{equation}
	\begin{equation}\label{z}
	z=\frac{{(RN^2)}^{\cdot}}{RN^2} ~ ,
	\end{equation}
	with $\dot{f} \equiv df/d\phi$.
Then, \eqref{eq1}-\eqref{eq3} become
	\begin{eqnarray}\label{toreq1}
	\frac{dY}{dx}+3Y-\frac{{{e}^{-2x}}}{{{\nu }^{2}}}+\frac{\kappa ^{2}\left( {{\nu }^{2}}+2\left(V+\frac{\Lambda}{\kappa}\right) \right)}{2{{\nu }^{2}}\left( 1-2\xi {{\nu }^{2}}\kappa^{2} \right)}=0 ~ ,
	\end{eqnarray}
	\begin{eqnarray}\label{toreq2}
\sqrt{YZ}-\frac{\left( 2\xi {{\nu }^{2}}\kappa ^{2}-1 \right){{e}^{-2x}}}{{{\nu }^{2}}\left( 6\xi {{\nu }^{2}}\kappa ^{2}-1 \right)}+\frac{\kappa^{2}\left( {{\nu }^{2}}+\left(V+\frac{\Lambda}{\kappa}\right) \right)}{{{\nu }^{2}}\left( 6\xi {{\nu }^{2}}\kappa ^{2}-1 \right)}=0~,
	\end{eqnarray}
	\begin{eqnarray}\label{toreq3}
\sqrt{\frac{Y}{Z}}\frac{dZ}{dx}+Z+\frac{{{e}^{-2x}}}{{{\nu }^{2}}}+\frac{3\kappa^{2}\left( {{\nu }^{2}}+2\left(V+\frac{\Lambda}{\kappa}\right) \right)}{2{{\nu }^{2}}\left( 1-2\xi {{\nu }^{2}}\kappa^{2} \right)}=0~.
	\end{eqnarray}
After some computations, we show that equations \eqref{toreq1}-\eqref{toreq3} can be reduced to a highly nonlinear ordinary differential equation, which we shall refer to as the master equation for $Y(x)$, 

\begin{eqnarray}
  && \frac{{{d}^{2}}Y}{d{{x}^{2}}}+Y\frac{dY}{dx}\left\{ -\frac{3}{2}-\frac{3\tilde{\eta }}{4\hat{\eta }}+\frac{6}{\tilde{\eta }\hat{\eta }} \right\}+\frac{dY}{dx}\frac{{{e}^{-2x}}}{{{\nu }^{2}}}\left\{ \frac{1}{4}-\frac{{\hat{\eta }}}{{\tilde{\eta }}} \right\}+{{\left( \frac{dY}{dx} \right)}^{2}}\left\{ -\frac{1}{2}+\frac{{\hat{\eta }}}{{\tilde{\eta }}} \right\} \notag\\ 
 && +3\frac{dY}{dx}+{{Y}^{2}}\left\{ -\frac{9\tilde{\eta }}{4\hat{\eta }}+\frac{9\hat{\eta }}{{\tilde{\eta }}} \right\}+\frac{{{e}^{-2x}}}{{{\nu }^{2}}}Y\left\{ \frac{{\tilde{\eta }}}{{\hat{\eta }}}-\frac{3{{\nu }^{2}}\hat{\eta }}{{\tilde{\eta }}} \right\}+\frac{{{e}^{-4x}}}{{{\nu }^{2}}}\frac{{\hat{\eta }}}{4\tilde{\eta }}+\frac{{{e}^{-2x}}}{{{\nu }^{2}}}=0 \notag\\
\end{eqnarray}
where $\hat{\eta }$ and $\tilde{\eta}$ are constants defined as:
\begin{eqnarray}
  \hat{\eta }&=& 2\xi {{\nu }^{2}}{{\kappa }^{2}}-1 ~,\notag\\ 
  \tilde{\eta }&=&6\xi {{\nu }^{2}}{{\kappa }^{2}}-1  ~.
\end{eqnarray}
This master equation \eqref{mastereq} describes the $4$-dimensional static spacetimes of scalar-torsion theories with non-minimal derivative coupling. Once equation  \eqref{mastereq} has been solved, the scalar potential and the scalar field profile, can be determined. Specifically, the potential \( V \) is obtained from:
\begin{eqnarray}
    V=\frac{{{\nu }^{2}}\left( 1-2\xi {{\nu }^{2}}{{\kappa }^{2}} \right)}{{{\kappa }^{2}}}\left( -\frac{dY}{dx}-3Y+\frac{{{e}^{-2x}}}{{{\nu }^{2}}} \right)-\frac{\Lambda }{\kappa }-\frac{{{\nu }^{2}}}{2}
\end{eqnarray}
and the scalar field \( \phi \) is obtained from
\begin{eqnarray}
    \left( \frac{dx}{d\phi} \right)^2 = Y(x).
\end{eqnarray}
The next question is whether or not the solution of the differential equation in \eqref{mastereq} exists. In the next section, we prove that there is indeed local and global existence of the master equation in \eqref{mastereq}.


In this subsection we prove the local existence and uniqueness of the master equation using the Picard's iteration and the contraction mapping theorem. First, we define the dynamical quantity
	\begin{equation}
	{\bf{u}}   \equiv  \left( \begin{array}{c}
	Y  \\
	P_Y 
	\end{array} \right) ~ , \label{dynvar}
	\end{equation}
consider the interval \( I \equiv [x_0, x_0 + \varepsilon] \), where \( x_0 \in \mathbb{R} \) and \( \varepsilon \) is a small positive constant. Let \( P_Y \equiv \frac{dY}{dx} \) be defined on \( I \). Then, the master equation \eqref{mastereq} can be expressed as
	\begin{equation}
	\frac{d {\bf{u}} }{dx}  = \mathcal{J}({\bf{u}}, x)   \equiv  \left( \begin{array}{c}
	P_Y  \\
	J_Y
	\end{array} \right) ~ , \label{fungsiJ}
	\end{equation}
with
\begin{eqnarray}\label{JY}
  {{J}_{Y}}&=& Y{{P}_{Y}}\left\{ -\frac{3}{2}-\frac{3\tilde{\eta }}{4\hat{\eta }}+\frac{6}{\tilde{\eta }\hat{\eta }} \right\}+P_Y\frac{{{e}^{-2x}}}{{{\nu }^{2}}}\left\{ \frac{1}{4}-\frac{{\hat{\eta }}}{{\tilde{\eta }}} \right\}+P_{Y}^{2}\left\{ -\frac{1}{2}+\frac{{\hat{\eta }}}{{\tilde{\eta }}} \right\}+3{{P}_{Y}} \notag\\ 
 && +{{Y}^{2}}\left\{ -\frac{9\tilde{\eta }}{4\hat{\eta }}+\frac{9\hat{\eta }}{{\tilde{\eta }}} \right\}+\frac{{{e}^{-2x}}}{{{\nu }^{2}}}Y\left\{ \frac{{\tilde{\eta }}}{{\hat{\eta }}}-\frac{3{{\nu }^{2}}\hat{\eta }}{{\tilde{\eta }}} \right\}+\frac{{{e}^{-4x}}}{{{\nu }^{2}}}\frac{{\hat{\eta }}}{4\tilde{\eta }}+\frac{{{e}^{-2x}}}{{{\nu }^{2}}}  
\end{eqnarray}
\begin{lemma}\label{opJY}
Let $U \subset \mathbb{R}^2$ be an open set. The nonlinear operator $\mathcal{J}(\mathbf{u}, x)$ defined in \eqref{JY} is locally Lipschitz in $\mathbf{u}$ over $U$.
	\end{lemma}
\begin{proof}
We have the following estimate
\begin{eqnarray}\label{estJY}
   {{\left| {{J}_{Y}} \right|}_{U}}&\le& \left| Y{{P}_{Y}} \right|\left\{ \frac{3}{2}+\frac{3\tilde{\eta }}{4\hat{\eta }}+\frac{6}{\tilde{\eta }\hat{\eta }} \right\}+\left|P_Y\right|\frac{{{e}^{-2x}}}{{{\nu }^{2}}}\left\{ \frac{1}{4}+\frac{{\hat{\eta }}}{{\tilde{\eta }}} \right\}+\left| P_{Y}^{2} \right|\left\{ \frac{1}{2}+\frac{{\hat{\eta }}}{{\tilde{\eta }}} \right\}+3\left| {{P}_{Y}} \right| \notag\\ 
 &&+\left| {{Y}^{2}} \right|\left\{ \frac{9\tilde{\eta }}{4\hat{\eta }}+\frac{9\hat{\eta }}{{\tilde{\eta }}} \right\}+\frac{{{e}^{-2x}}}{{{\nu }^{2}}}\left| Y \right|\left\{ \frac{{\tilde{\eta }}}{{\hat{\eta }}}+\frac{3{{\nu }^{2}}\hat{\eta }}{{\tilde{\eta }}} \right\}+\frac{{{e}^{-4x}}}{{{\nu }^{2}}}\frac{{\hat{\eta }}}{4\tilde{\eta }}+\frac{{{e}^{-2x}}}{{{\nu }^{2}}}  	
\end{eqnarray} 
Since \( Y \) belongs to at least a \( C^2 \)-class of real functions, it is bounded on any closed interval \( I \). The nonlinear operator \( |\mathcal{J}(\mathbf{u}, x)|_U \) is bounded on \( U \), as both \( \hat{\eta} \) and \( \tilde{\eta} \) are smooth and constant. Furthermore, for any \( \mathbf{u}, \hat{\mathbf{u}} \in U \), the following holds:
\begin{eqnarray}\label{estJYLps}
{{\left| \mathcal{J}\left( u,x \right)-\mathcal{J}\left( \hat{u},x \right) \right|}_{U}}&\le & \left( Y\left| {{P}_{Y}}-{{{\hat{P}}}_{Y}} \right|+{{{\hat{P}}}_{Y}}\left| Y-\hat{Y} \right| \right)\left\{ \frac{3}{2}+\frac{3\tilde{\eta }}{4\hat{\eta }}+\frac{6}{\tilde{\eta }\hat{\eta }} \right\}+3\left| {{P}_{Y}}-{{{\hat{P}}}_{Y}} \right| \notag\\ 
 && +\left( {{P}_{Y}}+{{{\hat{P}}}_{Y}} \right)\left| {{P}_{Y}}-{{{\hat{P}}}_{Y}} \right|\left\{ \frac{1}{2}+\frac{{\hat{\eta }}}{{\tilde{\eta }}} \right\}+\left|P_Y-\tilde{P}_Y\right|\frac{{{e}^{-2x}}}{{{\nu }^{2}}}\left\{ \frac{1}{4}+\frac{{\hat{\eta }}}{{\tilde{\eta }}} \right\}\notag\\
 &&+\left( Y+\hat{Y} \right)\left| Y-\hat{Y} \right|\left\{ \frac{9\tilde{\eta }}{4\hat{\eta }}+\frac{9\hat{\eta }}{{\tilde{\eta }}} \right\} \notag\\ 
 && +\left| Y-\hat{Y} \right|\frac{{{e}^{-2x}}}{{{\nu }^{2}}}\left\{ \frac{{\tilde{\eta }}}{{\hat{\eta }}}+\frac{3{{\nu }^{2}}\hat{\eta }}{{\tilde{\eta }}} \right\}  ~.
\end{eqnarray} 
We can then establish the inequality
		\begin{equation}
		\left|  \mathcal{J}({\bf u}, x) -  \mathcal{J}( \hat{{\bf u}}, x) \right|_U  \le C_{ \mathcal{J}}(|\bf{u}|, |\hat{\bf{u}}|) | \bf{u} - \hat{\bf{u}}| ~ .  \label{localLipshitzcon}
		\end{equation}
which demonstrates that the operator $\mathcal{J}$ satisfies a local Lipschitz condition in ${\bf u}$. This proves the local existence of the master equation \eqref{mastereq}.
\end{proof}

In this subsection, we prove the existence of a regular global solution to \eqref{fungsiJ} defined on an interval $I \subset \mathbb{R}$. The equation can be expressed in integral form as
\begin{equation}
\mathbf{u}(x) = \mathbf{u}(x_0) + \int_{x_0}^{x} \mathcal{J}\big(\mathbf{u}(s), s\big) \, ds,
\label{IntegralEquation}
\end{equation}
which remains well-defined and finite for all $x \in I$.

To verify this claim, we introduce an extended interval $I^+ \equiv I^+_L \cup I^+_A$, where:
\begin{itemize}
    \item $I^+_L \equiv [x_0, L]$ is a finite subinterval, and
    \item $I^+_A \equiv (L, +\infty)$ represents the asymptotic regime,
\end{itemize}
with $L > |x_0|$ chosen sufficiently large. For any $\tilde{C} \in I^+_A$, we then consider\eqref{IntegralEquation}  as
\begin{equation}
{\bf{u} }(\tilde{C}) = {\bf{u} }(x_0) + \int_{x_0}^{L}\:\mathcal{J}\left( {\bf{u} }(s), s \right) ~ ds ~ +  \int_{L}^{\tilde{C}} ~ \mathcal{J}\left( {\bf{u} }(s), s \right) ~ ds . \label{IntegralEquation1}
\end{equation}
In the asymptotic region $I^+_A$, we have $P_Y = dY/dx \to 0$ and $Y \to Y_0$, causing $\mathcal{J}(\mathbf{u}(x), x) \to 0$ and ensuring the boundedness of the third term in \eqref{IntegralEquation1}. Since $Y$ is at least $C^2$-smooth, the integral equation \eqref{IntegralEquation} remains finite and well-defined on all of $I^+$. We similarly construct the negative interval $I^- = I^-_L \cup I^-_A$, where $I^-_L \equiv [-L, x_0]$ and $I^-_A \equiv (-\infty, -L)$, with $L > |x_0|$ chosen sufficiently large. In this case, analogous behavior emerges where
\begin{equation}
{\bf{u} }(-\tilde{C}) = {\bf{u} }(x_0) + \int_{-L}^{x_0}\:\mathcal{J}\left( {\bf{u} }(s), s \right) ~ ds ~ +  \int_{-\tilde{C}}^{-L} ~ \mathcal{J}\left( {\bf{u} }(s), s \right) ~ ds . \label{IntegralEquation2}
\end{equation}
Consider the asymptotic behavior on $I^-_A$ where $P_Y = dY/dx \to 0$ and $Y$ approaches a positive constant $Y_c$. In the limit $x \to -\infty$, the operator $J_Y$ diverges to infinity, which leads to the divergence of \eqref{IntegralEquation2}. Consequently, \eqref{IntegralEquation} admits no regular solution on the interval $I^-$. 
\begin{theorem}\label{theorem}
Consider the master equation \eqref{fungsiJ} defined on the interval $I_{x_0} \equiv [x_0, +\infty)$. For any finite $x_0 \in \mathbb{R}$, there exist well-defined solutions to \eqref{fungsiJ}.
\end{theorem}
\end{document}